\documentclass[reqno,12pt,letterpaper]{amsart}
\usepackage{amsmath,amssymb,amsthm,graphicx,mathrsfs,url}
\usepackage[usenames,dvipsnames]{color}
\usepackage[colorlinks=true,linkcolor=Red,citecolor=Green]{hyperref}
\usepackage{amsxtra}

\usepackage{sidecap}
\usepackage[small]{caption}

\def\?[#1]{\textbf{[#1]}\marginpar{\Large{\textbf{??}}}}

\def\smallsection#1{\smallskip\noindent\textbf{#1}.}
\let\epsilon=\varepsilon 

\newcommand{\RR}{{\mathbb R}}

\newcommand{\CC}{{\mathbb C}}
\newcommand{\ZZ}{{\mathbb Z}}

\newtheorem{theo}{Theorem}
\newtheorem{prop}{Proposition}[section]

\newtheorem{lemm}[prop]{Lemma}

\numberwithin{equation}{section}

\DeclareMathOperator{\Res}{Res}
\DeclareMathOperator{\Spec}{Spec}
\DeclareMathOperator{\comp}{comp}

\let\Im=\Imag
\DeclareMathOperator{\loc}{loc}

\let\Re=\Real

\DeclareMathOperator{\supp}{supp}
\DeclareMathOperator{\vol}{vol}

\DeclareMathOperator{\tr}{tr}
\newcommand{\CI}{C^\infty}
\newcommand{\CIc}{C^\infty_{\rm{c}}}

\def\indic{\operatorname{1\hskip-2.75pt\relax l}}

\title{A Fermi golden rule for quantum graphs}
\author{Minjae Lee}
\email{lee.minjae@math.berkeley.edu}
\author{Maciej Zworski}
\email{zworski@math.berkeley.edu}
\address{Department of Mathematics, University of California,
Berkeley, CA 94720, USA}

\begin{document}

\begin{abstract}
We present a Fermi golden rule giving rates of decay
of states obtained by perturbing  embedded 
eigenvalues of a quantum graph. To illustrate the procedure
in a notationally simpler setting we also present a Fermi Golden
Rule for boundary value problems on surfaces with constant curvature
cusps. We also provide a resonance existence
result which is 
uniform on compact sets of energies and metric graphs. The results are illustrated by 
numerical experiments.

\end{abstract}

\maketitle

\section{Introduction and statement of results}
\label{intr}

Quantum graphs are a useful model for spectral properties 
of complex systems.
The complexity is captured by the graph but analytic aspects remain
one dimensional and hence relatively simple. We refer to the
monograph by Berkolaiko--Kuchment \cite{BeKu}
for references to the rich literature on the subject. 

In this note we are interested in graphs with infinite leads
and consequently with continuous spectra. We study
dissolution of embedded eigenvalues into the continuum and 
existence of resonances close to the continuum. Our motivation 
comes from a recent Physical Review Letter \cite{gnu} by 
Gnutzmann--Schanz--Smilansky and from a mathematical study by 
Exner--Lipovsk\'y \cite{exlip}. 

We consider an oriented graph with vertices
$ \{v_j\}_{j=1}^J $, 
infinite leads $ \{ e_k\}_{ k=1}^{K} $, $ K > 0 $, and $ M$ 
finite edges $ \{e_m\}_{ m=K+1}^{M+K} $.
 We assume that each finite edge, $ e_m $, has two distinct vertices
as its boundary (a non-restrictive no-loop condition) and we 
write $ v \in e_m $ for these two vertices $ v $. An infinite lead
has one vertex. 
The set of (at most two)
common vertices of $ e_m $ and $ e_\ell $ is denoted by $ e_m \cap e_\ell $
and we we denote by $ e_m \ni v $ the set of all edges having $ v $ 
as a vertex.

The finite edges are assigned
length $ \ell_m $, $ K+1 \leq m \leq M +K $ 
and we put $ \ell_k = \infty $, $ 1 \leq k \leq K $, for the infinite
edges.  To obtain a {\em quantum graph} we define
a Hilbert space,
is given by 
\begin{gather*} L^2 := \bigoplus_{m=1}^{K+M} L^2 ( [0, \ell_m] )  ,  
\ \  L^2 \ni u = ( u_1,  \cdots 
u_{M+K} ) , \ \ u_m \in  L^2 ( [0, \ell_m] ) . \end{gather*}

We then consider the simplest quantum graph Hamiltonian which 
is unbounded operator $ P $ on $ L^2$ defined by
$ (P u)_m  = - \partial_x^2 u_m $ with 
\[  
\mathcal D ( P ) = \{ u : u_m \in H^2 ( [ 0 , \ell_m ] ) , \
u_m ( v ) = u_\ell ( v ) , \  v \in e_m \cap e_\ell , \ 
\sum_{ e_m \ni v } \partial_{\nu} u_m ( v ) = 0 \} . \]
Here $ \partial_\nu $ denotes the outward pointing normal at boundary
of $ e_v $: 
\[ u_m \in H^2 ( [ 0 , \ell_m ] ) , \ \ \partial_\nu u_m ( 0 ) = - u_m'(0), \ \
\partial_\nu u_m ( \ell_m ) = u_m'(\ell_m) .\]

Quantum graphs with infinite leads fit neatly into the general 
abstract framework of {\em black box} scattering \cite{SZ1} and hence we can 
quote general results \cite[Chapter 4]{res} in spectral and scattering theory.

\begin{SCfigure}
\includegraphics[width=7cm]{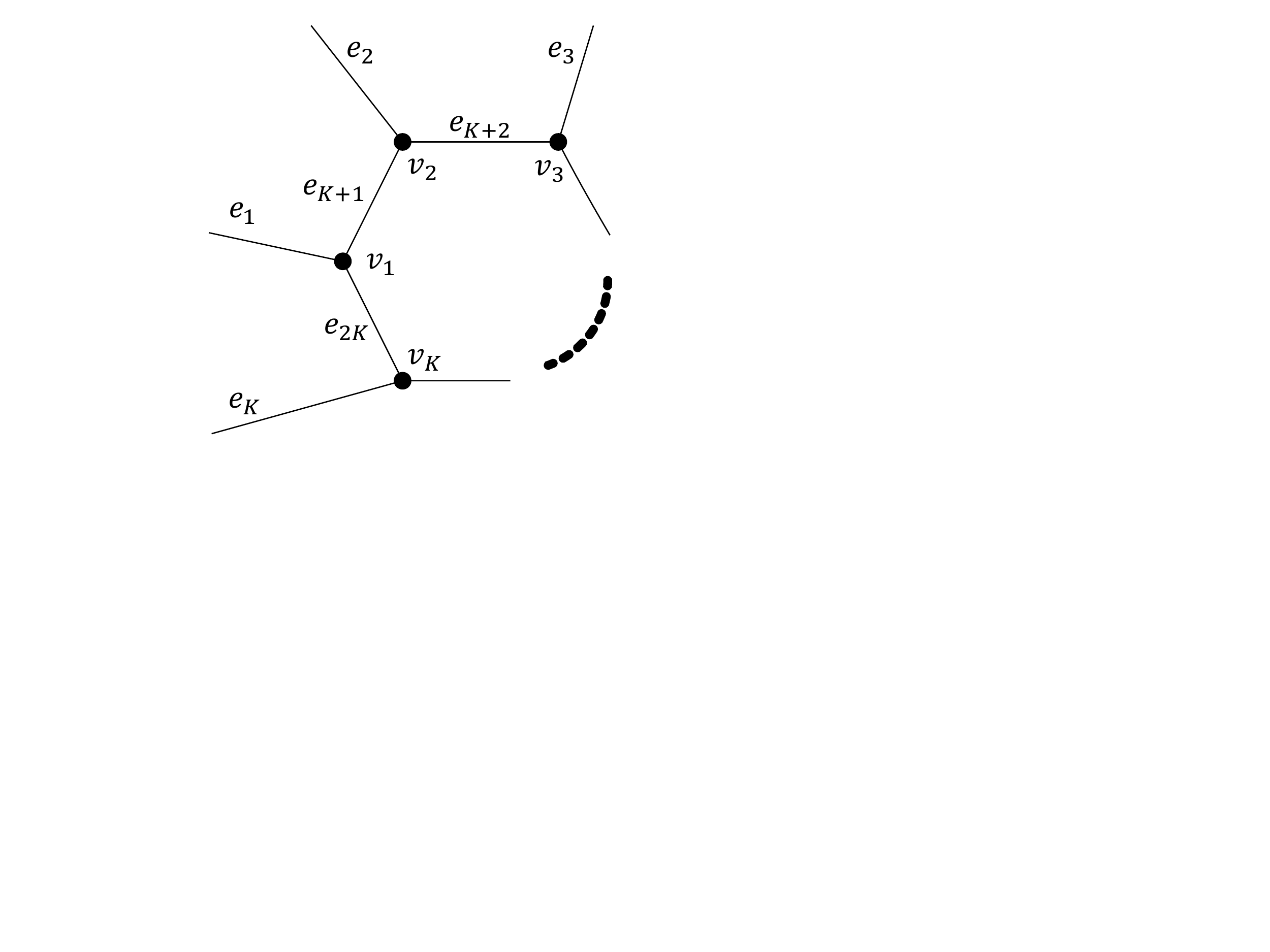}
\label{f:cycle}
\caption{A graph  given by a cycle
$ \{ e_{ k} \}_{k= K+1}^{2K} $ connected to $ K $ infinite leads $ \{ e_k \}_{k=1}^K $
at $ K $ vertices: $ v_k $, $ e_{K+k} \cap e_{ K+k-1} = v_k $, 
$ e_{ 2K } \cap e_{ K + 1} = v_1 $, $ e_{ k } \cap e_{ K+k} = v_k $.
The lengths of finite edges are given by $ \ell_k ( t ) = e^{-2 a_k ( t ) } \ell_k $, $ K + 1 \leq k \leq 2K $. If $ \ell_k ( 0 ) $'s are rationally
related then $ P ( 0 ) $ has eigenvalues, $ \lambda ( 0 ) $, embedded in the continuous spectrum. If $ \lambda ( 0 ) $ is simple then $ \lambda ( 0 ) $ 
belongs to a smooth family of resonances, $ \lambda ( t )$, 
$ \Im \lambda ( t) \leq 0 $.
Theorem \ref{t:0} and Example 1 in 
\S \ref{th0} show that in this case 
$ \Im \ddot \lambda = \lambda^2 \sum_{ k=1}^K | \langle \dot a u , 
e^k ( \lambda ) \rangle |^2 
$, where $ u $ is the normalized eigenfuction corresponding to $ u $ and
$ e^k ( \lambda ) $ is the generalized eigenfuction normalized in the $kth$ lead -- see \eqref{eq:ekla}.}
\end{SCfigure}

When $ K > 0 $ then the projection on the continuous spectrum of $ P $ is 
given in terms of generalized eigenfunctions $ e^k ( \lambda ) $, $ 1 \leq k \leq K $, which for $ \lambda \notin \Spec_{\rm{pp}} ( P ) $ are 
characterized as follows:
\begin{gather}
\label{eq:ekla} 
\begin{gathered}  e^k ( \lambda ) \in \mathcal D_{\rm{loc}} ( P ) ,  \ \ \
( P - \lambda^2 ) e^k (\lambda)  = 0 , \\  e^k_m ( \lambda , x ) = \delta_{mk} e^{ - i \lambda x } 
+ s_{ mk} ( \lambda ) e^{ i \lambda x } , \ \ \ 1 \leq m   \leq K. \end{gathered}
\end{gather}
The family $ \lambda \mapsto e^k ( \lambda ) \in \mathcal D_{\rm{loc} } ( P ) $
extends holomorphically to a neighbourhood of $ \RR $ and that defines
$ e^k ( \lambda ) $ for all $ \lambda $. We will in fact be interested in 
$ \lambda \in \Spec_{\rm{pp}} ( P ) $. The functions $ e^k $ parametrize
the continuous spectrum of $ P $ -- see \cite[\S 4.4]{res} and \eqref{eq:Rlaek} 
below. 

We now consider a family of quantum graphs obtained by varying the
lengths $ \ell_m $, $ K+1 \leq m \leq M+K  $:
\begin{equation}
\label{eq:amt}    \ell_m ( t ) = e^{ - a_m ( t )  } \ell_m ,  \ \ a_m ( 0 ) = 0 . 
\end{equation}
and the corresponding family of operators, $ P ( t ) $. 
The works \cite{exlip} and \cite{gnu} 
considered the case in which $ P ( 0 ) $ has embedded eigenvalues and
investigated
the resonances of the deformed family $ P ( t ) $ converging to these
eigenvalues as $ t \to 0 $. Here we present a Fermi golden rule type formula
(see \S \ref{cusps} for references to related mathematical work) 
which gives an infinitesimal condition for the disappearance of an 
embedded eigenvalue. It becomes a resonance of $ P $ and one can 
calculate the infinitesimal rate of decay. Resonances are defined
as poles of the meromorphic continuation of $ \lambda \mapsto ( P - \lambda^2)^{-1} $ to $ \CC $ as an operator $ L^2_{\comp} \to 
L^2_{ \loc } $ (see \cite[\S 4.2]{res} and for a self-contained
general argument Proposition \ref{p:reso}). We denote the set of 
resonances of $ P $ by $ \Res ( P ) $.

\begin{theo}
\label{t:0}
Suppose that $ \lambda^2 > 0 $ is a simple eigenvalue of $ P = P ( 0 ) $ and $ u $ is 
the corresponding normalized eigenfunction. Then for $ |t| \leq t_0 $ 
there exists a smooth function $ t \mapsto \lambda (t) $ such that 
$ \lambda ( t ) \in \Res ( P ) $ and
\begin{gather}
\label{eq:FGR0}
\begin{gathered} 
 \Im \ddot \lambda = - 
 \sum_{k=1}^K | F_k|^2, 
 \\ 
 F_k:= \lambda \langle \dot a u , e^k ( \lambda ) \rangle 
+ \lambda^{-1} \sum_{ v } \sum_{ e_m \ni v } 
{\textstyle \frac14} \dot a_m 
( 3 \partial_\nu u_m ( v)  \overline{ e^k ( \lambda , v )} -
u ( v ) \partial_\nu \overline { e^k_m ( \lambda, v ) } )
\end{gathered} \end{gather}
\end{theo}

\begin{figure}
\includegraphics[width=\textwidth]{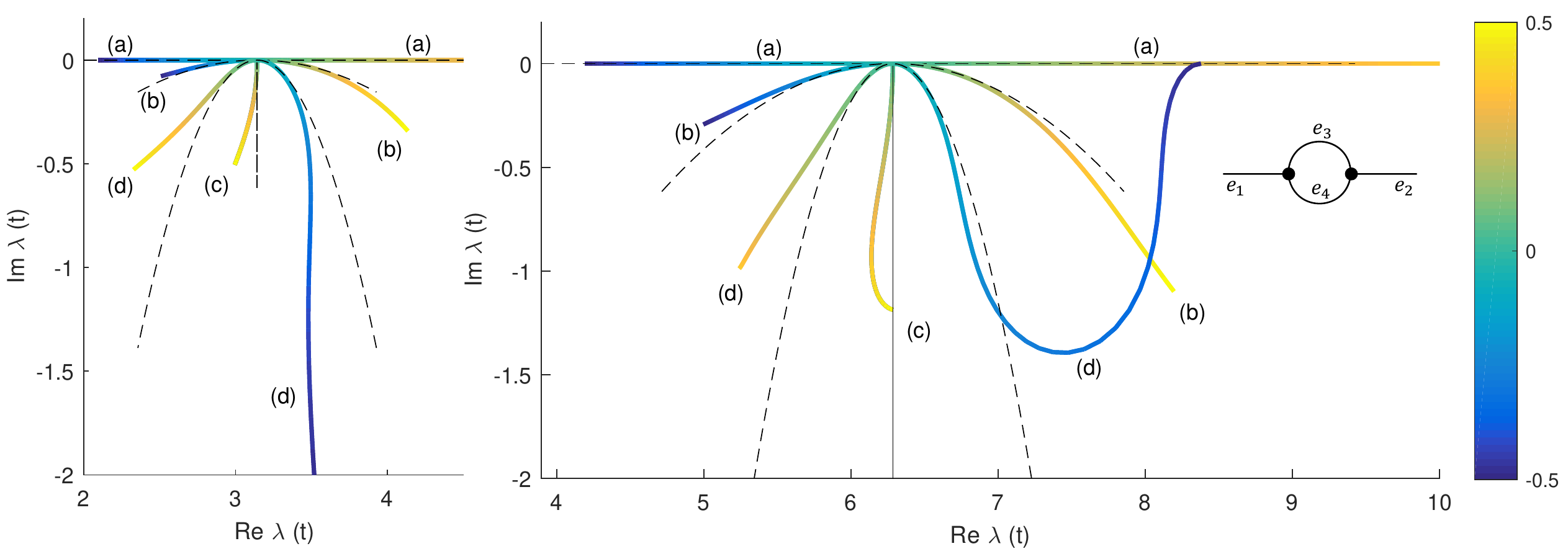}
\centering
\caption{\label{f:2}
A simple graph with embedded eigenvalues, $M=K=2$. Solid lines and dashed lines 
indicate the trajectory of $\lambda(t)$ and of the second order approximation $ 
\tilde \lambda ( t ) = \lambda + t \dot \lambda  + 
\frac i2 t^2 \Im \ddot \lambda $, respectively. (The colour coding indicates
the parameter $ t $ shown in the colour bar.) We approximate the real 
part linearly using \eqref{eq:zdot1} and the imaginary quadratically 
using \eqref{eq:FGR0}. The four cases are
(a): \(\ell_3 (t)=1-t,~\ell_4(t)=1-t\), 
(b): \(\ell_3 (t)=1-t,~\ell_4(t)=1\),
(c): \(\ell_3 (t)=1-t,~\ell_4(t)=1+t\), 
(d): \(\ell_3 (t)=1-t,~\ell_4(t)=1+2t\).
 }
\end{figure}

The proof is given in \S \ref{th0} and that section is concluded
with two examples: the first gives graphs and eigenvalues for 
which $ F_k = \lambda \langle \dot a u , e^k ( \lambda ) \rangle $
-- see Figures \ref{f:cycle} and \ref{f:2}. The second example
gives a graph and an eigenvalue for which the boundary terms in 
the formula for $ F_k $ are needed -- see Fig.~\ref{f:ex2}.

The formula \eqref{eq:FGR0} gives a condition for the existence 
a resonance with a nontrivial imaginary part (decay rate) near
an embedded eigenvalue of the unperturbed operator: 
$ D ( \lambda_0, c t ) \cap \Res ( P (t ) ) \neq \emptyset$ for some 
$ c $ and for $ |t| \leq t_0 $, where the constants $ c $ and $ t_0 $
depend on $ \lambda_0 $ and $ P ( t ) $. However,
it is difficult to estimate the speed with which 
the resonance $ \lambda ( t ) $ moves -- that is already visible 
in comparing Fig.~\ref{f:2} with Fig.~\ref{f:ex2}. 
(A striking example is given by $ P ( t ) = - \partial_x^2 + t V ( x ) $
where $ V \in \CIc ( \RR ) $ and $ t \to 0 $; infinitely many resonances
for $ t \neq 0 $ \cite{Zw} disappear and $ P ( 0 ) $ has only one resonance at $ 0 $.) Also, the result is not uniform if we vary $ \lambda_0 $ or the
lengths of the edges. 

The next theorem adapts the method of Tang--Zworski \cite{tz} and Stefanov 
\cite{St} (see
also \cite[\S 7.3]{res}) to obtain existence of resonances 
near any approximate eigenvalue and in 
particular near an embedded eigenvalue -- see the example following the statement.
In particular this applies to the resonances studied
in \cite{exlip} and \cite{gnu}. The method applies however to very 
general Hamiltonians -- for semiclassical operators on graphs the 
general black box resuls of \cite{tz} and \cite{St} apply verbatim.
The point here is that the constants are uniform even though the 
dependence on $ t $ is slightly weaker. 

To formulate the result we define $ D ( \lambda_0, r ) = \{ \lambda \in \CC : 
| \lambda - \lambda_0 | < r \} $ and 
\begin{equation}
\label{eq:HR} 
  \mathcal H_R  := \bigoplus_{ m = 1}^{ K} L^2 ( [ 0 , R ] ) \oplus
\bigoplus_{ m=K+1}^{K+M}  L^2 ( [ 0 , \ell_m ]) . \end{equation}

\begin{theo}
\label{t:q2r}
Suppose that $ P $ is defined above and the lengths, $ \ell_m $,  
have the property that $ \ell_m \in \mathcal L $, $ K+1 \leq m \leq M + K $
where $ \mathcal L $  is a fixed compact subset of
the the open half-line.

Then for any $ \mathcal L \Subset ( 0 , \infty ) $, $ I \Subset 
( 0 , \infty )  $, $ R > 0 $ and 
$ \gamma < 1 $  there exists $ \epsilon_0 > 0 $ such that 
\begin{equation}
\label{eq:quasim}
\exists \, u \in \mathcal H_R \cap \mathcal D (P) , \ \lambda_0 \in I 
\ \text{ such that } \ 
\|u \|_{L^2} = 1, \ \ 
\| ( P - \lambda_0^2) u \| = \epsilon < \epsilon_0 
\end{equation} 
implies
\begin{equation}
\label{eq:q2r}
\Res ( P ) \cap D ( \lambda_0 , \epsilon^\gamma ) \neq \emptyset .
\end{equation}
\end{theo}

\noindent
{\bf Example.} Suppose that $ P ( t ) $ is the family of operators
defined by choosing $ \ell_j = \ell_j ( t ) \in C^1 ( \RR ) $,  and that 
$ \lambda_0 > 0 $ is an 
eigenvalue of $ P ( 0 ) $. Then for any $ \gamma < 1 $
there exists $ t_0 $ such that for $ |t | \leq t_0 $
\begin{equation}
\label{eq:Ptq2r}  \Res ( P ( t) ) \cap D ( \lambda_0 , t^\gamma ) \neq \emptyset . \end{equation}

\begin{proof}
Let $ u^0  $ be a normalized eigenfunction of $ P ( 0 ) $ with 
eigenvalue $ \lambda_0 $; in particular $ u^0_k \equiv 0 $, $ 1 \leq k \leq K$. Choose $ \chi_j \in \CI ( \RR ; [ 0 , 1 ] ) $, 
$ j = 1,2 $, such that $ \chi_0 + \chi_1 = 1 $, $ \chi_j (s) = 1 $ near 
$ |j - s | < \frac13 $ and define $ u^-_m ( x) := \chi_0 ( x/\ell_m ) u^0_m ( x ) $ and $ u_m^+ ( x ) := \chi_1 ( x / \ell_m ) u^0_m ( x ) $, 
$ u^0 = u^+ + u^- $. 

We now define a quasimode for $ P ( t ) $, $ u = u (t ) $ needed 
in \eqref{eq:quasim}: 
\[  u_m ( t ) =  u_m^-  ( x )  +  u_m^+ ( x - \delta_m ( t ) ) , \ \ 
\delta_m ( t ) := \ell_m ( t ) - \ell_m ( 0 ) . 
\] 
For $ t $ small enough $ \supp u_m^- \subset [0 , \frac23) \subset
[ 0 , \ell_m ( t ) ) $
and $ \supp u_m^+ \subset ( \frac23 , \ell_m (0 ) ] 
\subset ( |\delta_m ( t )| , \ell_m ( 0 ) ] $. Hence
the values of $ u_m ( t ) $ and $ \partial_\nu u_m ( t) $ 
at the vertices are the same as those of $ u^0_m $ and 
 $ u_m ( t) \in \mathcal D ( P ( t ) )$. 
Also, since $ ( - \partial_x^2 - \lambda_0^2 ) u_m^0 = 0 $ and $ 
\chi_0^{(k)} = - \chi_1^{(k)} $, (and putting $ \ell_m = \ell_m ( 0 ) $)
\[ \begin{split} 
 [ ( P ( t ) - \lambda_0^2 ) u ( t ) ]_m & = 
\ell_m^{-2} ( \chi_0'' ( (x - \delta_m ( t )) /\ell_m  ) 
 u_m^0 ( x - \delta_m ( t ) )  - \chi_0'' ( x /\ell_m ) u_m^0 ( x )  )\\
 & \ \ \ \ \ \ +  2 \ell_m^{-1} ( 
   \chi_0' ( ( x - \delta_m( t ) ) / \ell_m ) u_m^0 ( x- \delta_m ( t ) ) 
 - \chi_0' ( x / \ell_m ) u_m^0 ( x ) ) .
  \end{split} \]
We note that all the terms are supported in $ ( \frac13 - |\delta_m ( t)| , 
\frac23 + |\delta_m ( t ) |) $ and elementary estimates show that
$ \| ( P ( t ) - \lambda_0^2 ) u ( t ) \| \leq C t $. For instance,
\[ \begin{split} 
\| \chi_0'' ( x ) ( u_m^0 ( x - \delta_m ( t )  ) - u_m^0 ( x ) ) 
\| & \leq C' | \delta_m ( t ) | \max_{ | x - \frac12| \leq
\frac16 + | \delta_m ( t) }| \partial_x u^0_m ( x ) | 
 \\
& \leq C' | \delta_m ( t ) | ( \| - \partial_x^2 u_0^m  \|_{L^2 ( ( \frac14, 
\frac34) ) } + \| u_0^m \|_{L^2 ( ( \frac14, \frac34) ) } ) \\
& \leq C'' ( \lambda_0^2 + 1 ) t. \end{split} \]
From \eqref{eq:q2r} we conclude (after decreasing $ \gamma $ and $ t_0 $)
that \eqref{eq:Ptq2r} holds. 
\end{proof}

\noindent
{\bf Remarks.} 1. A slightly sharper statement than \eqref{eq:q2r} can 
already be obtained from the proof in \S \ref{s:q2r}. It is possible
that in fact $ \Res ( P ) \cap D ( \lambda_0 , C_0 \epsilon  ) $
where $ C_0 $ depends on $ \mathcal L, R$ and $ \delta$. That
is suggested by the fact that the converse to this stronger conclusion is valid -- see Proposition \ref{p:converse}. This improvement would require 
finer
complex analytic arguments. It is interesting to 
ask if methods more specific to quantum graphs, in place of our general 
methods, could produce this improvement.

\noindent
2. By adapting Stefanov's methods \cite{St} one can strengthen the 
conclusion by adding adding a statement about multiplicities (see 
also \cite[Exercise 7.1]{res}) but again we opted for a simple presentation.

\smallsection{Acknowledgements}
 We are grateful for the support of 
National Science Foundation under the grant DMS-1500852. We would
also like to thank Semyon Dyatlov for helpful discussions and assistance
with figures.

\section{A Fermi golden rule for boundary value problems: surfaces with cusps}
\label{cusps}

To illustrate the Fermi golden rule in the setting of boundary value
problems we consider surfaces, $ X $, with cusps of constant negative curvature. 
That means that $ ( X , g ) $ is a surface with 
a smooth boundary and a decomposition (see Fig.~\ref{f:cusp})
\begin{gather}
\label{eq:X0X1}
\begin{gathered} 
  X = X_1 \cup X_0 , \ \ \partial X_0 = \partial X_1 \cup \partial X ,  \ \ 
\partial X_1 \cap \partial X_0 = \emptyset , \\
(X_1 , g|_{X_1} ) \simeq ( [ a, \infty )_r \times (\RR/ \ell \ZZ)_\theta , 
dr^2 + e^{-2r} d\theta^2 ) . 
\end{gathered}
\end{gather}
We consider the following family of unbounded operators on $ L^2 ( X )$:
\begin{gather}
\label{eq:Poft}
\begin{gathered}  P (t ) = - \Delta_g - {\textstyle{\frac14}} , 
\ \ \
\mathcal D ( P ( t ) ) = 
\{ u \in H^2 ( X ) : \partial_\nu u |_{\partial X } = \gamma ( t ) u |_{\partial X }
\} .
\end{gathered}
\end{gather}
where $ t \mapsto \gamma ( t ) \in \CI ( \partial X ) $ is a smooth 
family of functions on $ \partial X $ and $ \partial_\nu $ is the outward 
pointing normal derivative. The spectrum
of the operator $ P $ has the following well known decomposition:
\begin{gather*} \Spec ( P ) = \Spec_{\rm{pp}}( P ) \cup \Spec_{\rm{ac} } ( P ) , \ \ 
\Spec_{\rm{ac} } ( P ) = [ 0 , \infty ) , \\ \Spec_{\rm{pp} } ( P ) = 
\{ E_j \}_{ j=0}^{J} , \ \ -\textstyle{ \frac14} \leq E_0 < E_1 \leq E_2 \cdots ,  
\ \ 0 \leq J \leq + \infty . \end{gather*}
(When $ J = + \infty $ then $ E_j \to \infty $.) The eigenvalues $ E_j > 0 $
are {\em embedded} in the continuous spectrum. In addition the resolvent
$R ( \lambda ) := ( P -\lambda^2 )^{-1} : L^2 \to L^2  $, 
$ \Im \lambda > 0 $, has a meromorphic continuation to $ \lambda \in \CC $
as an operator $ R ( \lambda ) : \CIc ( X ) \to \CI ( X ) $. Its poles 
are called {\em scattering resonances}. Under generic
perturbation of the metric in $ X_0 $ all embedded eigenvalues become 
resonances. For proofs of these well known facts see \cite{cdv} and also
\cite[\S 4.1 (Example 3), \S 4.2 (Example 3), \S 4.4.2]{res} 
for a presentation from the point
of view of {\em black box scattering} \cite{SZ1}. 

\begin{SCfigure}
\includegraphics[width=3.3in]{fgqg.1}
\caption{\label{f:cusp}
A surface with one cusp end and a boundary. Suppose we consider 
a family of boundary conditions for the Laplacian $ - \Delta $: 
$ \partial_\nu w = \gamma ( t ) w $ at $ \partial X $. 
The Laplacian has continuous spectrum with a family of generalized
eigenfuctions $ e ( \lambda ) \in \CI ( X ) $ -- see \eqref{eq:Eis}.
Suppose that for $ t = 0 $, $ \lambda^2 $ is a {\em simple embedded}
eigenvalue of $ - \Delta $ with the boundary condition 
$ \partial_\nu w = \gamma ( 0 ) w $, with the normalized eigenfunction
given by $ u $. Then $ \lambda = \lambda ( 0 ) $ belong to a smooth
family of {\em resonances} of Laplacians with boundary 
condition $ \partial_\nu w = \gamma ( t ) w $, and $ 
\Im \ddot \lambda = -\frac{1}{ 4 \lambda^2} | \langle \dot \gamma u , 
e ( \lambda ) \rangle|^2 $ -- see Theorem \ref{t:1}.}
\end{SCfigure}

The generalized eigenfunctions, $ e ( \lambda, x ) $, 
describing the projection onto the continuous
spectrum have the following properties:
\begin{gather}
\label{eq:Eis}
\begin{gathered}
( P - \lambda^2 ) e ( \lambda , x ) = 0 , \ \ \ \ \frac{ 1 } {\ell} \int_0^\ell 
e( \lambda , x )|_{X_1 } d \theta = e^{\frac r 2 }\left( e^{ - i \lambda r } + s ( \lambda ) e^{ i \lambda r } \right) , \\
( R ( \lambda ) - R ( - \lambda ) ) f = {\textstyle{\frac{i}{2 \lambda} } } e ( \lambda , x ) \langle f , e ( \lambda , \bullet ) \rangle, \ \ \lambda \in \RR , \ \ f \in \CIc ( X ) ,
\end{gathered}
\end{gather}
see \cite[Theorem 4.20]{res}. With these preliminaries in place 
we can now prove
\begin{theo}
\label{t:1}
Suppose that the operators $ P ( t ) $ are defined by \eqref{eq:Poft}
and that $ \lambda > 0 $ is a {\em simple} eigenvalue of $ P ( 0 ) $ and
$ ( P ( 0 )- \lambda^2  ) u = 0 $, $ \| u  \|_{ L^2 } = 1 $. 

Then there exists a smooth function $ t \mapsto \lambda ( t ) $, $ |t | < t_0 $, such
that $ \lambda ( 0 ) = \lambda $, $ \lambda ( t ) $ is a scattering resonance of $ P ( t ) $ and 
\begin{gather}
\label{eq:FGR1}
\begin{gathered}
\Im \ddot \lambda = - \frac{ \! 1 }{ 4 \lambda^2 } \left| \langle \dot \gamma u , e \rangle_{ 
L^2 ( \partial X ) } \right|^2 , \  \   
 \ e ( x ) = e( \lambda, x ) , 
\end{gathered}
\end{gather}
where $ e ( \lambda, x ) $ is given in \eqref{eq:Eis}, 
$ \dot f := \partial_t f|_{ t=0 } $ and $ L^2 ( \partial X ) $ is defined using
the metric induced by $ g $.
\end{theo}

\noindent
{\bf Remarks.} 1. For recent advances in mathematical study of the Fermi golden rule
in more standard settings of mathematical physics and for numerous references see Cornean--Jensen--Nenciu \cite{cjn}.

\noindent
2. In the case of scattering on constant curvature surfaces with cusps
the Fermi golden rule was explicitly stated by Phillips--Sarnak -- see\cite{phisa} and for a recent discussion \cite{petr}.
For a presentation from the black box point of view see \cite[\S 4.4.2]{res}.

\noindent
3. The proof generalizes immediately to the case of several cusps (which is 
analogous to a quantum graph with several leads), $ ( X_k, g|_{X_k} ) 
\simeq ( [a_k, \infty ) \times \RR/\ell_k \ZZ , dr^2 + e^{-2r} d\theta^2 $, 
$ 1 \leq k \leq K $.
In that case the generalized eigenfunction are normalized using 
\[  \frac{ 1 } {\ell_m} \int_0^{\ell_m} 
e^k ( \lambda , x )|_{X_m } d \theta = e^{\frac r 2 }\left( \delta_{ km } 
e^{ - i \lambda r } + s_{km}  ( \lambda ) e^{ i \lambda r } \right) . \]
The Fermi golden rule for the boundary value problem \eqref{eq:Poft} is
given by 
\begin{equation}
\label{eq:FGR2}
\Im \ddot \lambda = - \frac{ \! 1 }{ 4 \lambda^2 } \sum_{ k=1}^K
\left| \langle \dot \gamma u , e^k \rangle_{ 
L^2 ( \partial X ) } \right|^2 , \  \   
 \ e^k ( x ) = e^k ( \lambda, x ) .
\end{equation}

\begin{proof}
For notational simplicity we assume that $ \gamma( 0 ) 
\equiv 0 $, that is that $ P ( 0 ) $ is the Neumann Laplacian on $ X $. We will
also omit the parameter $ t $ when that is not likely to cause confusion. It
is also convenient to use $ z = \lambda^2 $ and to write $ \langle \bullet, 
\bullet \rangle $ for the $ L^2 ( X , d\vol_g ) $ inner product and 
$ \langle \bullet, \bullet \rangle_{ L^2 ( \partial X ) } $ for the inner
product on $ L^2 ( \partial) $ with the measure induced by the metric $ g $.

We first define the following orthogonal projection:
\begin{gather}
\label{eq:indyk} 
\begin{gathered}  \indic_{ r \geq R } u := \frac{1}{\ell} \int_0^\ell u|_{ X_1 \cap 
\{ r \geq R \}} \,  d\theta , \ \ \ \ \indic_{ r \geq  R} : L^2 ( X ) \to L^2 ( [ R, \infty )
 , e^{-r } dr ) ,  \ \
R > a , \\ \indic_{ r \leq R } := I - \indic_{ r \geq R } , \ \ 
\mathcal H_R := \indic_{r \leq R } L^2 ( X ) . 
\end{gathered}
\end{gather}
The smoothness of scattering resonances arising from a smooth perturbation
of a simple resonance follows from smooth dependence 
of the continuation of $ ( P ( t ) - \lambda^2 )^{-1} $ (see 
Proposition \ref{p:smooth} below for a general argument). Let $ t \mapsto 
u ( t ) $, $ u ( 0 ) = u $ denote a smooth family of resonant states:
\begin{gather}
\label{eq:Ptz}
\begin{gathered}   ( P ( t ) - z (t ) ) u ( t ) = 0 , \ \ \ \frac{1 } { \ell } \int_0^\ell
u ( t ) |_{ X_1 } d\theta = a ( t ) e^{ \frac r 2} e^{ i \lambda ( t ) r } , \\ 
a ( 0 ) = 0 , \ \ \Im \lambda ( t ) \leq 0 , \ \ \lambda(0)^2 = z(0) .   
\end{gathered}
\end{gather}
The second equation in \eqref{eq:Ptz} means that $ u( t ) $ is {\em outgoing}
-- see \cite[\S 4.4]{res}.

The self-adjointness of $ P ( t ) $ and integration by parts for the zero mode
in the cusp show that for $ u = u ( t )$ and $ P = P ( t ) $, 
\begin{equation}
\label{eq:Imz}
\begin{split} 0 & =  \Im \langle ( P - z ) u , \indic_{ r \leq R } u \rangle \\
& = 
- \Im \partial_r ( \indic_{ r \geq R } u ) ( R ) \overline{ \indic_{ r \geq R } u } ( R ) - \Im z \| \indic_{ r \leq  R} u \|_{L^2 ( X ) }^2  . \end{split} \end{equation}
(See \cite[(4.4.17)]{res} for a detailed presentation in the
general black box setting.) Since $ \Im \dot z = 0 $ (as $ \Im z ( t ) \leq 0 $, see
also \eqref{eq:zdot} below) and since $ \indic_{ r \geq R } u ( 0 ) = 0 $, we have
have, at $ t = 0 $, 
$  \Im \ddot z = - 2 \Im \partial_r ( \indic_{ r \geq R } \dot u ) ( R ) \overline{ \indic_{ r \geq R } \dot u } $.
We would like to argue as in \eqref{eq:Imz} but in reverse. However, as 
$ \dot u $ will not typically be in $ \mathcal D ( P ) $ we now obtain boundary 
terms:
\begin{equation}
\label{eq:zddot}  \Im \ddot z =  2 \Im \langle ( P - z ) \dot u , \indic_{r \leq  R} \dot u \rangle
+ 2 \Im  \langle \partial_\nu \dot u , \dot u \rangle_{ L^2 ( \partial X ) }. 
\end{equation}
We now need an expression for $ \dot u$. Since 
$ ( P ( t ) - z ( t ) ) u ( t ) = 0 $, $ \partial_\nu u |_{\partial X } = \gamma u|_{
\partial X } $, we have (at $ t = 0 $),
\begin{equation}
\label{eq:Pminz}  ( P - z ) \dot u = \dot z u  , \ \ \partial_\nu \dot u |_{ \partial X } 
= \dot \gamma u |_{\partial X } . \end{equation}
In addition, differentiation of the second condition in \eqref{eq:Ptz} shows
that $ \dot u $ is outgoing. 

Without loss of generality we can assume that $ u = u ( 0 ) $ is real valued.
Choose $ g \in \bar{C}^\infty ( X, \mathbb R ) $ (real valued, compactly supported and smooth up to the 
boundary) such that $ \partial_\nu g |_{\partial X } = \dot \gamma u |_{\partial X } $.
We claim that
\begin{equation}
\label{eq:orthg} 
\langle \dot z u  - ( P - z ) g , u \rangle = 0 .
\end{equation}
In fact, Green's formula shows that the left hand side of \eqref{eq:orthg}
is equal to $ \dot z  + \int_{\partial X } \dot \gamma  u ^2 $. On the other
hand, using the fact that $ \indic_{ r \leq R } u ( 0 ) = u ( 0 ) $, 
\begin{equation}
\label{eq:zdot}  \begin{split} 
0 & = - \frac{d}{dt} \langle ( P ( t) - z ( t)  ) u ( t) , \indic_{ r \leq R } u ( t) \rangle|_{ t = 0 }  = 
 \langle  \dot z u  - ( P - z )  \dot u ,  u \rangle \\
 & =  \dot z  + \int_{\partial X } \partial_\nu \dot u u = \dot z  + \int_{\partial X } \gamma u^2 . \end{split} \end{equation}

In view of \eqref{eq:orthg}, $ v := g + R ( \lambda ) ( \dot z 
- ( P - z ) g ) $, 
$ \lambda^2 = z $, $ \lambda > 0 $, is well defined, outgoing (see \eqref{eq:Ptz})
and solves the boundary value problem \eqref{eq:Pminz} satisfied by $ \dot u $.
Since the eigenvalue at $ z $ is simple that means that $ \dot u - v $ is a multiple of 
$ u $ (see \cite[Theorem 4.18]{res} though in this one dimensional case this is 
particularly simple). Hence
\begin{equation}
\label{eq:dotu}
\dot u = \alpha u + g + R ( \lambda ) ( \dot z u - ( P - z ) g ) . 
\end{equation}
With this formula in place we return to \eqref{eq:zddot}. First we note that
the first term on the right hand side vanishes: 
\begin{equation}
\label{eq:first} \begin{split} 
\Im \langle ( P - z ) \dot u , \indic_{ r \leq R } \dot u \rangle & = 
\Im \langle \dot z u , \dot u \rangle 
= \dot z \Im \langle u , \alpha u + g + R ( \lambda ) ( \dot z u  - ( P - z ) g ) \rangle 
\\
& = \dot z \Im \alpha + \dot z \Im \langle u , R ( \lambda ) ( \dot z u - ( P - z ) g ) \rangle \\
& = \dot z \Im \alpha  . 
\end{split}\end{equation}

Here we used the fact that $ u $ and $ g $ were chosen to be real. The last
identity followed from \eqref{eq:orthg}. To analyse the second term on the right hand side of 
\eqref{eq:zddot}  we recall some properties of the Schwartz kernel of the resolvent:
\begin{equation}
\label{eq:Reso}
R ( \lambda ) ( x, y ) = R ( \lambda ) ( y , x ) = \overline{ R ( - \overline \lambda ) ( x , y )} , \ \ \lambda \in \CC . \end{equation}
(The first property follows from considering $ \lambda = i k $, $ k \gg 1 $, 
and using the fact that $ \overline { P u } = P \bar u $, and the second from 
considering $ \Im \lambda \gg 1 $, $ z = \lambda^2 $, and noting that $ ( ( P - z )^{-1})^* = ( P - \bar z)^{-1} $.)
Using \eqref{eq:zddot},\eqref{eq:Pminz},\eqref{eq:first},\eqref{eq:dotu},\eqref{eq:Reso},\eqref{eq:zdot} and the fact that $ u $ and $ g $ are real, 
we now see that
\begin{equation}
\label{eq:zddot20} \begin{split}
\Im \ddot z & = 2 \dot z \Im \alpha + 
2 \Im \langle \dot \gamma u , \dot u \rangle_{ L^2 (\partial X) } 
\\ &  = 2 \dot z \Im \alpha + 2 \Im \alpha \langle \dot \gamma u , u \rangle + 
2 \Im \langle \dot \gamma u ,  [ R ( \lambda ) ( \dot z u - ( P - z ) g ) ]|_{\partial X } 
 \rangle_{L^2 ( \partial X ) } \\
& = \textstyle{\frac 1  i} \langle \dot \gamma u , [ ( R ( \lambda ) - R ( - \lambda ) ) ( \dot z u - ( P - z ) g ) ]|_{\partial X } \rangle _{ L^2 ( \partial X )}
. \end{split} \end{equation}
Since $ ( R ( \lambda ) - R ( - \lambda ) ) u = 0 $ we have now use \eqref{eq:Eis}
to see that 
\[  \begin{split}
[ ( R ( \lambda ) - R ( - \lambda ) ) ( \dot z u - ( P - z ) g ) ]|_{\partial X } & = 
-  {\textstyle{\frac i {2 \lambda} } } e( \lambda )|_{ \partial X } 
\int_{ X } \overline{ e ( \lambda )} ( P - z) g \\ 
& = - {\textstyle{\frac i {2 \lambda} } } e( \lambda )|_{ \partial X } 
\int_{\partial X } ( \partial_\nu {\overline e( \lambda ) }g - \partial_\nu g 
\overline{ e ( \lambda ) } ) 
\\
& =  {\textstyle{\frac i {2 \lambda} } } e( \lambda )|_{ \partial X } \langle 
\dot \gamma u , e \rangle_{ L^2 ( \partial X ) } . \end{split} \]
Inserting this into \eqref{eq:zddot20} gives \eqref{eq:FGR1} completing the proof.\end{proof}

\section{Proof of Theorem \ref{t:0}}
\label{th0}

We follow the same strategy as in the proof of Theorem \ref{t:1} but 
with some notational complexity due to the graph structure. 

Let $   H^2 := \bigoplus_{m=1}^{M+K} H^2 ( [ 0 , \ell_m ]) $. Then 
for $ u , v \in H^2 $, $ ( \partial_x^k u)_m := \partial_x^k u_m $, 
\begin{equation}
\label{eq:parts}  \begin{split}  - \langle 
\partial_x^2 f, g \rangle_{L^2 } & = 
\langle 
\partial_x f , 
\partial_x g \rangle_{L^2} - \sum_{ v} \sum_{ e_m \ni v }
\partial_\nu f_m( v ) \bar g_m  ( v )  \\
& = - \langle  f , 
\partial_x^2 g \rangle_{L^2} 
+ \sum_{v} \sum_{e_m \ni v } 
\big(  f_m ( v ) \partial_\nu \bar g_m ( v ) - \partial_\nu f_m ( v ) \bar g_m ( v ) \big).
\end{split}
\end{equation}
We note here that the sum over vertices can be written as a sum over
edges:
\begin{equation}
\label{eq:edgev} \sum_{v } \sum_{ e_m \ni v } \partial_\nu f_m ( v ) \bar g_m ( v ) = \sum_{m=1}^{M+K} \sum_{ \, v \in \partial e_m } \partial_\nu f_m ( v ) \bar g_m ( v ) . 
\end{equation}

Just as in \S \ref{cusps} the domain of the deformed operators will 
change but we make a modification which will keep the Hilbert space on which 
$ \widetilde P ( t ) $ (we change the notation from \S \ref{intr} and 
will use $ P ( t ) $ for a unitarily equivalent operator) acts fixed by changing the lengths in \eqref{eq:amt}. 
For that 
let 
\begin{gather*}  L_t^2 := \bigoplus_{m=1}^{M+K} L^2 ( [0, e^{ - a_m ( t) } \ell_m] ) , \ \   L^2 :=   L^2_0 , \ \ U ( t) :   L_t^2 \to   L^2 ,
\\    [ U ( t ) u ]_m ( y ) := 
e^{ - a_m ( t ) /2 } u_m ( e^{ -a_m(t) } y ) , \ \ 
U(t)^{-1} = U(t)^* .  \end{gather*}
Let $ \widetilde P ( t ) $ be defined in $ L^2_t $ by $ ( \widetilde P (t)  u)_m  = - \partial_x^2 u_m $, 
\[  
\mathcal D ( \widetilde P ( t)  ) = \{ u : u_m \in H^2 ( [ 0 , e^{-a_j ( t )} \ell_m ] ) , \
 u_m ( v ) =  u_\ell ( v ) , \  v \in e_m \cap e_\ell , \ 
\sum_{ e_m \ni v } \partial_{\nu} u_m ( v ) = 0 \} . \]
That is just the family of Neumann Laplace operators on the graph with 
the lengths $ e^{-a_j ( t)  } \ell_j $. 

On $ L^2 $ we define a new family of operators: $ P ( t) := U(t) \widetilde P ( t ) 
U(t)^* $. It is explicitly given by $ [ P ( t ) u ]_m = -e^{2 a_m(t) } \partial_x^2 u_m $, 
\begin{equation}
\label{eq:Poft} \begin{split}
\mathcal D ( P ( t)  ) = \{ u \in H^2 : \; & 
e^{ a_m ( t ) / 2 } u_m ( v ) = e^{ a_\ell ( t ) / 2 } u_\ell ( v ) , 
\    v \in e_m \cap e_\ell ,
\\
& \ \ \ \ 
\sum_{ e_m \ni v } e^{3 a_m(t)/2 } \partial_{\nu} u_m ( v ) = 0 
\} . \end{split} \end{equation}

Using Proposition \ref{p:smooth} from the next section we see
that for small $ t $ there exists a smooth family 
$ t \mapsto u ( t ) \in H^2_{\rm{loc}} $ such that 
\begin{gather}
\label{eq:Ptz1}
\begin{gathered}    ( P ( t) - z( t) ) u ( t ) = 0 , \ \ u_k ( t , x ) = 
a ( t) e^{ i \lambda ( t) x } , \ \ k \geq M+ 1,  \\ \Im \lambda (t ) \leq 0 , 
\ \ \lambda ( 0 )^2 = z , \  \ \lambda ( 0 ) >0 . 
\end{gathered}
\end{gather}
We defined $ \mathcal H_R $ by  \eqref{eq:HR} and denote by 
$ \indic_{ x \leq R }  $ the orthogonal projection $ L^2 \to \mathcal H_R $.

Writing $ P  = P ( t ) $, $ u = u ( t ) $, $ z = z ( t )$ 
we see, as in \eqref{eq:Imz}, that 
\begin{equation}
\label{eq:Imz1} \begin{split} 
0 & = \Im \langle ( P  - z ) u  , \indic_{ x \leq R } u \rangle 
 = - \Im  \sum_{ m=1}^{K} \partial_x u_m ( R ) \bar u_m (  R  ) 
- \Im z \| u \|_{ \mathcal H_R}^2 . \end{split}  \end{equation}
 We recall that $ e_m$, $ 1 \leq m \leq K $ are the infinite edges with unique boundaries.
Hence, using \eqref{eq:parts},  at $ t = 0 $, 
\begin{equation}
\label{eq:zddot2}
\begin{split} \Im \ddot z & = 2 \Im \sum_{ m=1}^{K} \partial_x \dot u_m ( R ) \overline{\dot u}_m (  R ) \\
& = 2 \Im \langle ( P - z ) \dot u , \indic_{x \leq R } \dot u \rangle + 2 \Im 
\sum_{ v }\sum_{ e_m \ni v }
\partial_\nu {\dot u}_m ( v ) \overline{\dot  u}_m ( v ) .
\end{split} 
\end{equation}
We now look at the equation satisfied by $ \dot u $ at $ t = 0 $:
\begin{equation}
\label{eq:dotPt}   \frac{d}{dt} ( P ( t ) - z ( t ) ) u ( t ) = 
2 \dot a ( - \partial_x^2 u )  - \dot z u + ( P - z ) \dot u = 
( 2 \dot a z - \dot z ) u + ( P - z ) \dot u . \end{equation}
Hence, 
\begin{gather}
\label{eq:dotu2} 
\begin{gathered}  ( - \partial_x^2 - z ) \dot u_m = ( \dot z + 2 z \dot a_m ) u_m 
 , \ \ \ \sum_{ e_m \ni v } \partial_\nu \dot u_m ( v ) = - {\textstyle{\frac32}}  \sum_{ e_m \ni v } \dot a_m 
\partial_\nu u_m ( v ) , \\
 \dot u_m ( v ) - \dot u_\ell ( v )  = {\textstyle{\frac12}} 
( \dot a_\ell  - \dot a_m ) u ( v ) , \ \ v \in e_m \cap e_\ell
  . 
\end{gathered}
\end{gather}
We used here the fact that $ u(v) := u_m ( v ) $ does not depend on $ m $.
The second condition can be formulated as $ \dot u_m ( v ) = w ( v ) - 
\frac 12 \dot a_m ( v ) u ( v )$, where $ w := \partial_t ( e^{ a(t)/2 } u ( t ) )
|_{t=0} $ is continuous on the graph.  

To find an expression for $ \dot u $ (similar to \eqref{eq:dotu}) we first find
\[  g \in \bigoplus_{ m = 1}^K \CIc ( [ 0 , \infty ) )  \oplus \bigoplus_{ m = K + 1 }^{M+K} \CI ( [ 0 , \ell_m ] )
 , \]
such that 
\begin{equation}
\label{eq:gioconda} \sum_{ e_m \ni v } \partial_\nu  g_m ( v ) = - 
{\textstyle \frac 32}
\sum_{ e_m \ni v } \dot a_m u_m ( v ) , \ \  g_m ( v ) - g_\ell ( v ) = 
{\textstyle \frac 12} ( \dot a_\ell - \dot a_m ) u ( v ) . 
\end{equation}
We can assume without loss of generality that both $ g $ and $ u $ are real 
valued.

In analogy to \eqref{eq:orthg} we claim that
\begin{equation}
\label{eq:orthg1}
\langle ( \dot z - 2 z \dot a ) u - ( P - z ) g , u \rangle = 0 . 
\end{equation}
In fact,
using  \eqref{eq:parts}, \eqref{eq:dotPt} and \eqref{eq:dotu2} we obtain 
\begin{equation}
\label{eq:zdot1}  \begin{split} 
0 & = - \frac{d}{dt} \langle ( P ( t) - z ( t)  ) u ( t) , \indic_{ x \leq R } u ( t) \rangle|_{ t = 0 }  = 
 \langle  \dot z u - 2 z \dot a  u - ( P - z )  \dot u ,  u \rangle \\
 & =  \dot z  - 2 z \langle \dot a u , u \rangle + 
 \sum_{ v } \sum_{ e_m \ni v } ( \partial_\nu \dot u_m ( v ) u ( v ) - 
 \dot u_m ( v ) \partial_\nu u_m ( v ) ) \\
 & = \dot z - 2 z \langle \dot a u , u \rangle + 
 \sum_{ v } \sum_{ e_m \ni v } ( - {\textstyle \frac 32} \dot a_m \partial_\nu u_m ( v ) u ( v ) - 
 ( w ( v ) - {\textstyle \frac 12} \dot a_m u ( v ) ) \partial_\nu u_m ( v ) ) \\
 & = \dot z - 2 z \langle \dot a u , u \rangle - \sum_v \sum_{ e_m \ni v }
 \dot a_m \partial_\nu u_m ( v ) u ( v ) .
     \end{split} \end{equation}
(We used the continuity of $ u $ and the Neumann condition 
$ \sum_{e_m \ni v } \partial_\nu u_m ( v ) = 0 $.) Since $ g $ and $ u $ 
satisfy the same boundary conditions \eqref{eq:dotu2} and \eqref{eq:gioconda} 
\eqref{eq:orthg1} follows from \eqref{eq:zdot1}. 

As in the derivation of \eqref{eq:dotu} we now see that for some $ \alpha \in \CC $
we have
\begin{equation}
\label{eq:dotu1}  \dot u = \alpha u + g + R ( \lambda ) ( \dot z u - 2 z \dot a u - ( P - z ) g ) .\end{equation}
With this in place we return to \eqref{eq:zddot2}. The first term on the 
right hand side is 
\begin{equation*}
 \begin{split} 
2 \Im \langle ( P - z ) \dot u , \indic_{ x \leq R } \dot u \rangle & = 
2 \Im \langle \dot z u - 2 z \dot a u , \dot u \rangle 
\\
& 
= 2 \Im \langle \dot z  u - 2 z \dot a u , \alpha u + g + R ( \lambda ) ( \dot z u 
- 2 z \dot a u - ( P - z ) g ) \rangle 
\\
& = 2\Im \alpha ( \dot z - 2 z \langle \dot a u , u \rangle ) \\
& \ \ \ \ \ \ \  
- 4 z \Im \langle 
\dot a u , R ( \lambda ) ( \dot z u - 2 z \dot a u - ( P - z ) g ) \rangle . \\
\end{split}\end{equation*}
(We used here the simplifying assumption that $ g $ and $ u $ are real valued.)

As in \eqref{eq:zddot20} we conclude that
\begin{equation*}
\label{eq:zddot22}\begin{split}  4 z \Im \langle 
\dot a u , R ( \lambda ) ( - \dot z u + 2 z \dot a u + ( P - z ) g ) \rangle  
& = 
  {\textstyle{\frac {2 z}  i } } \langle \dot a u , [ ( R ( \lambda ) - R ( - \lambda ) ] ( 2 z \dot a u + ( P - z ) g ) \rangle 
\end{split} \end{equation*}
Now, as in \eqref{eq:Eis}, \cite[Theorem 4.20]{res} shows that 
\begin{equation}
\label{eq:Rlaek}  ( R ( \lambda ) - R ( - \lambda ) ) f = {{\frac{i}{2 \lambda} } } \sum_{ k=1}^K  e^k ( \lambda , x ) \langle f , e^k ( \lambda , \bullet ) \rangle, \ \ \lambda \in \RR , \ \ f \in \mathcal H_R   ,
\end{equation}
which means that (with $ z = \lambda^2 $ and $ e^k = e^k ( \lambda ) $)
\[ \begin{split} &  {\textstyle{\frac { 2z} i } } \langle \dot a u , [ ( R ( \lambda ) - R ( - \lambda ) ] ( 2 z \dot a u + ( P - z ) g ) \rangle  = \\ 
& \ \ \ \ \ \ \ \ - 2 {  \lambda^3 } 
\sum_{ k=1}^K |\langle \dot a u , e^k  \rangle |^2 -
  \lambda  \sum_{ k=1}^K \langle \dot a u , e^k\rangle \langle 
e^k ( \lambda ) , ( P - z ) g \rangle . 
\end{split} \]
The second term on the right hand side is now rewritten using \eqref{eq:parts}
and the boundary conditions \eqref{eq:gioconda}:
\[ \begin{split} & \lambda \sum_{k=1}^K \langle \dot a u , e^k \rangle
\left( \sum_{ v } \sum_{ e_m \ni v } ( \partial_\nu e^k_m ( v ) g_m ( v ) - 
\partial_\nu g_m ( v ) e^k ( v ) ) \right) = \\
& \ \ \ \ \ \
 \lambda \sum_{k=1}^K \langle \dot a u , e^k  \rangle
\left( \sum_{ v } \sum_{ e_m \ni v }
{\textstyle\frac12} \dot a_m 
 ( - \partial_\nu e^k_m ( v )  u( v ) +
3 \partial_\nu u_m ( v ) e^k ( v ) ) \right) . \end{split}
 \]
We conclude that
\begin{equation}
\label{eq:ffirst}
\begin{split} 
& 2 \Im \langle ( P - z ) \dot u , \indic_{ x \leq R } \dot u \rangle = 
 2 \Im \alpha ( \dot z - 2 z \langle \dot a u , u \rangle ) 
 - 2 \lambda^3 \sum_{ k=1}^K |\langle \dot a u , e^k \rangle |^2 \\
 & \ \ \ \ \ \ \
   - 2 \lambda \sum_{k=1}^K \langle \dot a u , e^k  \rangle
\left( \sum_{ v } \sum_{ e_m \ni v }
{\textstyle\frac14} \dot a_m 
 (  3 \partial_\nu u_m ( v ) e^k ( v )  - \partial_\nu e^k_m ( v )  u( v )  ) \right) 
.
\end{split}
\end{equation}
A similar analysis of the second term on the right hand side of 
\eqref{eq:zddot} shows that 
\begin{equation}
\label{eq:fffirst}
\begin{split}
& 2 \Im \sum_{ v }\sum_{ e_m \ni v } \partial_\nu {\dot u}_m ( v ) \overline{\dot  u}_m ( v ) = \Im \alpha 
\left( 2 \sum_v \sum_{ e_m \ni v }
 \dot a_m \partial_\nu u_m ( v ) u ( v ) \right) \\ 
 & \ \ \ - 
 2 \lambda^{-1}  \sum_{ k=1}^K \left| \sum_{ v } \sum_{ e_m \ni v }
{\textstyle\frac14} \dot a_m 
 ( \partial_\nu e^k_m ( v )  u( v ) - 
3 \partial_\nu u_m ( v ) e^k ( v ) )  \right|^2 \\
& \ \ \ - 2 \lambda \sum_{ k=1}^K \langle \overline {  \dot a u, e^k } \rangle
\left( \sum_{ v } \sum_{ e_m \ni v }
{\textstyle\frac14} \dot a_m 
 ( 3 \partial_\nu u_m ( v ) \overline e^k ( v ) - \partial_\nu \overline e^k_m ( v )  u( v )  ) \right) .
\end{split}
\end{equation}
Inserting \eqref{eq:ffirst},\eqref{eq:fffirst} into \eqref{eq:zddot2}, 
using \eqref{eq:zdot1} and $ \Im \ddot z = 2 \lambda \Im \ddot \lambda $ gives
\eqref{eq:FGR0}.
\qed

\noindent
{\bf Example 1.} 
Consider a connected graph with \(M\) bonds and \(K\) leads. Suppose that an embedded eigenvalue \(\lambda\) is simple and satisfies 
\begin{equation}\label{eq:ex_cond}
\lambda \ell_m \in \pi \mathbb{Z},~m=1,\cdots,M.
\end{equation}  
Then 
\begin{equation}\label{eq:ex_thm}
\Im \ddot \lambda = -\sum_{k=1}^K \left| \lambda\langle \dot a u, e^k (\lambda) \rangle \right|^2.
\end{equation}
\begin{proof}  
\(u_m(x)=C_m \sin(\lambda x)\) where \(e_m\) and a lead are meeting at a vertex. Since the graph is connected, \(u_m(x)=C_m \sin(\lambda x)\) for \(1\le m \le M\). Let
\(n_m=\frac{\lambda \ell_m}{\pi}\) and let \[e^k_m(\lambda,x) = A_{mk}\sin(\lambda x) + B_{mk} \cos (\lambda x) .\]
Then $ u_m(\ell_m) =u_m(\ell_m)=0 $ and 
\[ 
\partial_\nu u_m (0) = (-1)^{n_m+1}\partial_\nu u_m (\ell_m) , \
e_m^k (\lambda,\ell_m) = (-1)^{n_m} e_m^k (\lambda,0). 
\]
We can use this and 
\eqref{eq:edgev} to reduce \(F_k\) in \eqref{eq:FGR0} to 
\begin{align*}
F_k& = 
\lambda \langle \dot a u , e^k ( \lambda ) \rangle 
+ \lambda^{-1} \sum_{ m=1 }^M {\textstyle \frac34} \dot a_m 
( \partial_\nu u_m ( 0)  \overline{ e^k ( \lambda , 0 )} +\partial_\nu u_m ( \ell_m)  \overline{ e^k ( \lambda , \ell_m )} )
\\&=\lambda \langle \dot a u , e^k ( \lambda ) \rangle .
\end{align*}
Theorem \ref{t:0} then gives \eqref{eq:ex_thm}. \end{proof}

\noindent
{\bf Example 2.} Let us consider a graph with $ M =5, K = 2 $ and 
four vertices: see in Fig.~\ref{f:ex2}.  
Let \(\ell_m(0)=1,~ 1\le m \le 5\). Then the sequence of 
embedded eigenvalues \(\lambda\) is given as \(S_1 \cup S_2\) where
\[S_1=\pi\mathbb{Z}, \quad S_2= \left\{\lambda  \; :\; \tan\lambda +2\tan {\textstyle{\frac{\lambda}{2}}} =0, \quad \lambda\notin {\textstyle{\frac{\pi}{2} }} \mathbb{Z}\right\}.\]
If \(\lambda \in S_1\), then \eqref{eq:ex_cond} is satisfied. If \(\lambda \in S_2\), however, we have (with $ v_1 $ and $ v_2 $ corresponding to $ x = 0 $ 
for $ e_3 $, $ e_6 $ and $ e_4 $, $e_5 $ respectively, and $ v_4 $ to $ x = 0 $
for $ e_7 $)
\[u_3(x)=C\sin(\lambda x),\quad u_4(x)=C \sin(\lambda x),\quad u_5(x)=-C\sin(\lambda x),\]
\[u_6(x)=-C\sin(\lambda x),\quad u_7(x)= C\frac{\sin \lambda}{\sin \frac{\lambda}{2}} \sin\left(\lambda \left(x-{\textstyle{\frac{1}{2}}}\right)\right), \] where \(C>0\) is the normalization constant.
Note that \[u(v_3)=C\sin\lambda\ne 0, \quad u(v_4)=-C\sin\lambda \ne 0.\] So we do not have the simple formula \eqref{eq:ex_thm} in this case. 

\begin{figure}
\includegraphics[width=\textwidth]{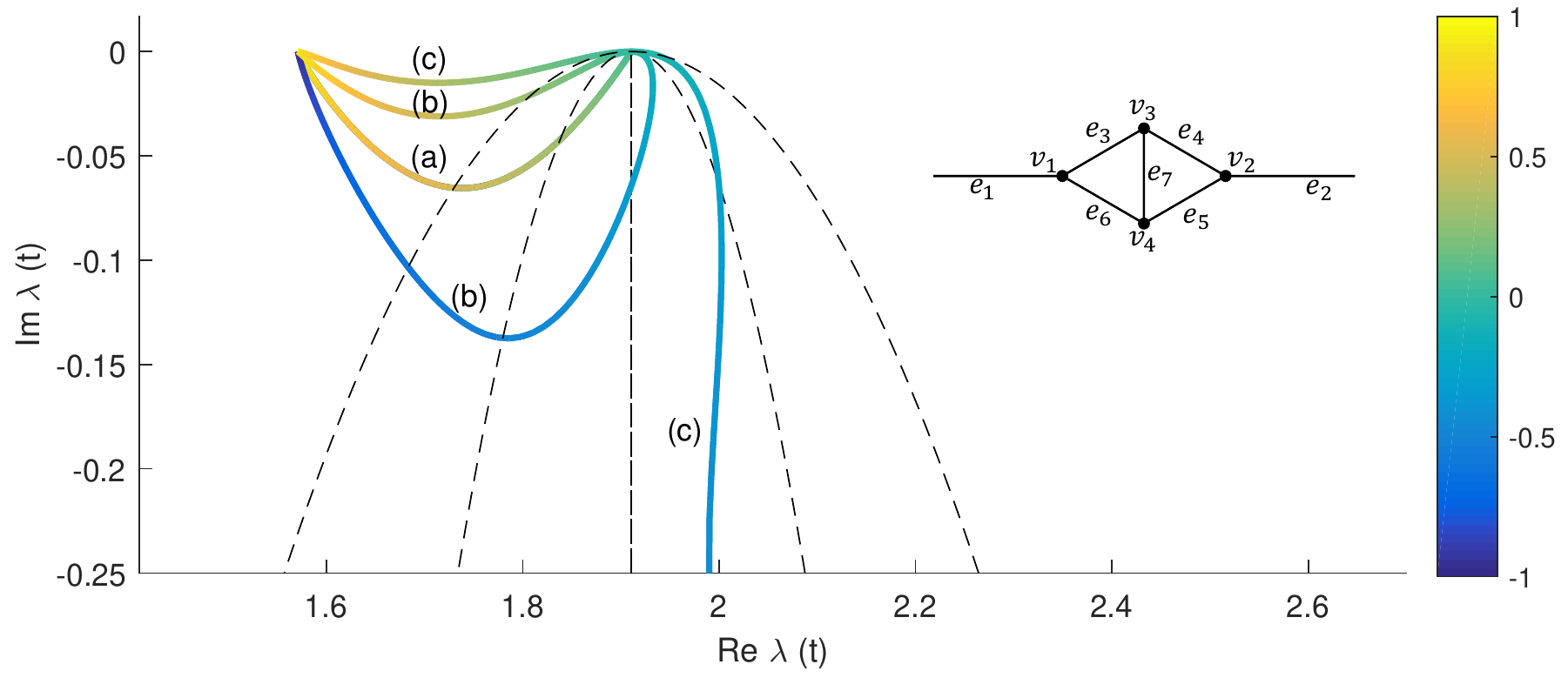}
\centering
\caption{\label{f:ex2}
The graph from Example 2: in this case boundary terms in 
our Fermi golden rule appear at some embedded eigenvalues
such as $ \lambda_0$ which is the smallest solution of 
$ \tan \lambda + 2 \tan \frac{\lambda} 2 = 0 $, $ 
\lambda_0 \approx 1.9106 $. We consider the following
variation of length: $ \ell_3=1-t, \ell_4=1+t, \ell_5=1-t, \ell_6=1+t $, 
and (a): $ \ell_7=1 $ (b): $ \ell_7=1+t/2 $ (c): $ \ell_7=1+t $.}
\end{figure}

\section{Proof of Theorem \ref{t:q2r}}
\label{s:q2r}

The proof adapts to the setting of quantum graphs and of quasimodes $ u$
satisfying \eqref{eq:quasim} the arguments of \cite{tz}. They have 
origins in the classical work of Carleman \cite{Car}
on completeness of eigenfunctions for classes non-self-adjoint operators, see also \cite{St} and \cite{StVo}.
 
We start with general results which 
are a version of the arguments of
\cite[\S 7.2]{res}. In particular they apply without modification to 
quantum graphs with general Hamiltonians and general boundary conditions.
We note that for metric graphs considered here much
more precise estimates are obtained by Davies--Pushnitski \cite{DaPu}
and Davies--Exner--Lipovsk\'y \cite{exda} but since we want uniformity 
we present an argument illustrating the black box point of view.

\begin{prop}
\label{p:reso}
Suppose that $ P$ satisfies the assumptions of 
Theorem \ref{t:q2r} and $ \Omega_1 \Subset \Omega_2 \Subset \CC $,
where $ \Omega_j $ are open sets. 

Then there exist constants $ C_1 $ depending only on $ \Omega_2$ and $ \mathcal L $, 
and $ C_2 $ depending on $ \Omega_1, \Omega_2 $, $ R $ and $ \mathcal L $ such that 
\begin{gather}
\label{eq:res_bound}
\begin{gathered}
| \Res ( P )\cap \Omega_2 | \leq C_1 ,\\
\| \indic_{r \leq R } R ( \lambda ) \indic_{r \leq R } \|_{ L^2 \to L^2}
\leq C_2  \prod_{\zeta \in \Res ( P )\cap \Omega_2  } | \lambda - \zeta |^{-1} , \ \ \lambda \in \Omega_1 ,
\end{gathered}
\end{gather}
where the elements of $ \Res ( P ) $ are included according to their multiplicities.
\end{prop}
\begin{proof} 
Let 
$  R_0 ( \lambda  ) : \bigoplus_{ k=1}^K L^2_{\comp} ( e_k ) \to 
\bigoplus H^2_{\rm{loc}} \cap H_{0,\rm{loc}} ^1  ( e_k)  $,
be defined as the diagonal operator acting on each component as $ R^0_0 ( \lambda ) $, the Dirichlet resolvent on $ L^2_{\comp} ( 
[ 0 , \infty ) ) $ continued {\em analytically} to all of $ \CC $:
\[ R^0_0 ( \lambda ) f ( x ) = 
 \int_0^\infty \frac{ e^{ i \lambda ( x + y ) } - e^{ i \lambda | x - y | } }
{ 2 i \lambda}   f ( y ) dy . \]
To describe $ \indic_{r \leq R }  R ( \lambda ) \indic_{r \leq R }  $
we follow the general argument of \cite{SZ1} (see also \cite[\S 4.2,4.3]{res}).
For that 
we choose $ \chi_j \in \CIc $, $ j = 0 , \cdots, 3 $ to be 
equal to $ 1 $ on all edges and to satisfy 
\begin{gather*}
\chi_j |_{ e_k} \in \CIc ( [ 0, 2 R ) ) , \ \ \chi_0 |_{ e_k} ( x ) =1 , \  
 x\leq R , \ \ 
\chi_j|_{ e_k } ( x )  =  1, \  x  \in \supp \chi_{ j-1}|_{e_k } , 
\end{gather*}
for  $ k =1 , \cdots, K $. For $ \lambda_0 $ with $ \Im \lambda_0 > 0 $,
we define
\[
Q ( \lambda, \lambda_0 ) := ( 1- \chi_0 ) R_0 ( \lambda ) ( 1 - \chi_1 ) 
+ \chi_2 R ( \lambda_0 ) \chi_1 , \ \  Q ( \lambda, \lambda_0 ) : L^2_{\comp} 
\to \mathcal D_{\loc} ( P ) . \]
Then 
\begin{gather*} ( P - \lambda^2 ) Q ( \lambda, \lambda_0 ) = I + K ( \lambda, \lambda_0 ) , \\
K_0 ( \lambda, \lambda_0 ) := - [ P , \chi_0 ] R_0( \lambda ) ( 1 - \chi_1 ) 
+ ( \lambda_0^2 - \lambda^2 ) \chi_2  R ( \lambda_0 ) \chi_1 + 
[ P , \chi_2 ] R ( \lambda_0 ) \chi_1 . \end{gather*}
We now choose $ \lambda_0 = e^{ \pi i /4 } \mu $, 
$ \mu \gg 1 $. Then 
\begin{equation}
\label{eq:Klala0} 
\text{ $ I + K_0 ( \lambda_0 , \lambda_0 ) \ $ and 
$ \  I + K_0 ( \lambda_0 , \lambda_0 ) \chi_3 \  $ are invertible on $ \ L^2 $, }
\end{equation}
$ K ( \lambda , \lambda_0 ) \chi_3 $ is compact, and 
\begin{equation} 
\label{eq:resform}
R ( \lambda ) = Q ( \lambda, \lambda_0 ) ( I + K_0( \lambda, \lambda_0 ) \chi_3)^{-1} ( I - K_0 ( \lambda , \lambda_0 ) ( 1 - \chi_3) ) ,
\end{equation}
where $ \lambda \mapsto ( I + K_0( \lambda, \lambda_0 ) \chi_3)^{-1} $ is
a meromorphic family of operators.
We now put 
\[  K ( \lambda, \lambda_0 ) := K_0( \lambda, \lambda_0 ) \chi_3 \]
and conclude that
\begin{equation} 
\label{eq:resform}
\indic_{r \leq R }  R ( \lambda ) \indic_{r \leq R }  = \indic_{r \leq R }  Q ( \lambda, \lambda_0 ) \chi_3 ( I + K( \lambda, \lambda_0 ) )^{-1} \indic_{r \leq R } ,
\end{equation}
and the set of resonances is given by the poles of $ ( I +  K ( \lambda, \lambda_0 ) )^{-1} $. (See \cite[\S 4.2]{res} and in particular 
\cite[(4.2.19)]{res}.)

We now claim that $ K ( \lambda , \lambda_0 ) $ is 
of trace class for $ \lambda \in \CC $ 
and that for a any compact 
subset   $ \Omega \Subset \CC $ there exists a constant $ C_3 $ depending
only on $ \Omega $, $ \mathcal L $ and $ \lambda_0  $ such that
\begin{equation}
\label{eq:Klala0}    \| K ( \lambda, \lambda_0 )\|_{\tr} \leq C_3 . 
\end{equation}
To see this, let $ \widetilde P $ be the operator of $ \mathcal H_{3R} $ where
we put, say the Neumann boundary condition at $ 3 R $ on each infinite
lead. Let $ \widetilde P_{\min}, \widetilde P_{\max} $ 
be the same operators but  on metric graphs were
all the length $ \ell_j \in \mathcal L $, $ K +1 \leq j \leq K + M $ 
were replaced by $ \ell_{\min} := \min\mathcal L $ and 
$ \ell_{\max} :=\max \mathcal L $ respectively. 
These operators have discrete spectra and the ordered
eigenvalues of these operators satisfy
\begin{equation}
\label{eq:brack}   \lambda_p ( \widetilde P_{\max} ) \leq \lambda_p ( \widetilde P ) 
\leq \lambda_p ( \widetilde P_{\min} ) . 
\end{equation}
This is a consequence of the following lemma:

\begin{lemm}
Suppose that the unbounded operator $ \widetilde P_k ( t ) : \mathcal H_{\rho} \to \mathcal H_{\rho }$, $ \rho > 0 $,  with edge length given by 
\[  \ell_k (t ) = \rho , \ 1 \leq k \leq K, \ \ 
  \ell_m (t) =  e^{-\delta_{mk}  t}  \ell_m , \ \ K+1 \leq m \leq M + K , \]
and 
\[ \mathcal D ( \widetilde P_k ( t ) ) = 
 \{ u : u_m \in H^2 ( [ 0 , \ell_m (t) ] ) , \
u_m ( v ) = u_\ell ( v ) , \  v \in e_m \cap e_\ell , \ 
\sum_{ e_m \ni v } \partial_{\nu} u_m ( v ) = 0 \} . \]
If $ 0 = \mu_0 ( t ) \leq \mu_1 ( t ) \leq \mu_2 ( t ) \cdots , $ 
is the ordered sequence of eigenvalues of $ P ( t ) $,  then 
$ \mu_p ( t ) $ is an increasing function of $ t $.
\end{lemm}
\begin{proof}
From \cite[Theorem 3.10]{BK1} we know that 
if $ \mu $ is an eigenvalue of $ P ( s ) $ of multiplicity $ N$
then we can choose analytic functions $ \mu^n ( t ) \in \RR $, $ u^n ( t ) \in \mathcal D ( P ( s ) ) $,   such that $ \mu^n ( s ) = \mu $, and 
for small $ t-s $, $ P ( t ) u^n ( t ) = \mu^n ( t ) u^n ( t ) $, and $ \{ u^n ( t ) \}_{ n=1}^N $
is an orthonormal setting spanning $ \indic_{ | P ( t ) - \mu | 
\leq \epsilon } L^2 $, for $ \epsilon > 0 $ small enough. 
The lemma follows from showing that $ \partial_t \mu^n ( s) \geq 0 $ 
for any $ n $. 

Without loss of generality we can assume that
$ s = 0$. 
We can then use the same calculation as in \eqref{eq:zdot1}
with $ z  = \mu^n ( 0 ) $, $ a_m ( t ) = \delta_{km} t $
and $ u = u^n ( 0 ) $. That gives
\[  \mu_p ' ( 0 ) = 2 \mu_p ( 0 ) \langle u, u \rangle_{ L^2 ( e_k )} 
+ \sum_{ v \in \partial e_k } \partial_\nu u_k ( v ) u_k ( v ) . \]
Since $ u_k ( x ) = a \sin \sqrt \mu_p x + b \cos \sqrt \mu_p x $ ,
for some $ a , b \in \RR $, a calculation shows that
\[  \mu_p ' ( 0 ) =  \mu_p ( 0 ) \ell_k ( a^2 + b^2 ) \geq 0 , \]
completing the proof.
\end{proof}

The inequality \eqref{eq:brack} follows from the lemma as 
we can change the length of the edges in succession. 
The Weyl law for $ \widetilde P $ (see \cite{BeKu}) and the
fact that $ \widetilde P \chi_3 = P \chi_3 $ (where $ \chi_3 $
denotes the multiplication operator), now shows that 
for any operator $ A :L^2 \to \mathcal D ( P )  $, 
\[ \| \chi_3 A \chi_3 \|_{\tr} 
\leq C_4 \| P \chi_3 A \chi_3 \| + C_4 \| \chi_3 A \chi_3 \| , \]
where the constant $ C_4 $ depends only on $ \mathcal L $. From this
we deduce \eqref{eq:Klala0} and $ \| \bullet \| = \| \bullet \|_{ L^2 
\to L^2 } $.  For instance,
\[ \begin{split}
  \| [ P , \chi_2]  R ( \lambda_0 ) \chi_1 \|_{\tr} & \leq 
C_4  \|  P [ P, \chi_2]  R ( \lambda_0 ) \chi_1 
\|+ C_4  \| [ P ,  \chi_2]  R ( \lambda_0 ) \chi_1  \|
\\
& = 
C_4  \| [ P , [ P, \chi_2 ] R ( \lambda_0 ) \chi_1 \|
+  C_4 ( 1 + |\lambda_0|^2 ) \| [ P , \chi_2]  R ( \lambda_0 ) \chi_1  \| \\
 & \leq C_5 .\end{split}
 \]
Here we used the facts that $ \chi_2 \equiv 1 $ on the support of $ \chi_1 $,
hence $ [ P , \chi_2 ] \chi_1 = 0 $, and that $ [ P, [ P , \chi_2] ] $ 
$ [ P ,\chi_2 ] $ are second and first order operators respectively 
and and that $ R ( \lambda_0 ) $ maps $ L^2 $ to $ \mathcal D ( P ) $. 
The other terms in $ K ( \lambda, \lambda_0 ) $ are estimated similarly
and that gives \eqref{eq:Klala0}. (Finer estimates for large $ \lambda $ are possible -- see \cite[\S 4.3, \S 7.2]{res} and \cite{tz}-- but we concentrate here on uniformity near a given energy.) 

Now, let $ \Omega_3 = \{ \lambda : | \lambda - \lambda_0 | < R $
where $ R $ is large enough so that
$ \Omega_2 \subset \Omega_3 $. 
It follows that for a constant $ C_3 $ depending only on $ \Omega_3 $ and
$ \mathcal L$, 
(and hence only on $ \Omega_2 $), we have 
\begin{equation}
\label{eq:detKla} | \det ( I + K ( \lambda , \lambda_0 ) | \leq e^{C_3} .
\end{equation}
(For basic facts about determinants see for instance \cite[\S B.5]{res}.)
Writing 
\[  ( I + K ( \lambda_0 , \lambda_0 )) ^{-1} = ( I - ( I + K ( \lambda_0, 
\lambda_0 ) )^{-1} K ( \lambda_0, \lambda_0 ) ) \]
we obtain
\begin{equation*}
\begin{split} | \det ( I + K ( \lambda_0 , \lambda_0 ) ) |^{-1} & =
| \det ( I + K ( \lambda_0 , \lambda_0 ) ^{-1} | \\
& \leq \exp \left( \| ( I + K ( \lambda_0 , \lambda_0 ) )^{-1} \| 
\| K ( \lambda_0 , \lambda_0 ) \|_{\tr} \right) \leq e^{C_4} , 
\end{split} \end{equation*}
that is
\begin{equation}
\label{eq:detKla0}
 | \det ( I + K ( \lambda_0 , \lambda_0 ) ) | \geq e^{ - C_4 } , 
 \end{equation}
where $ C_4 $ depends only on $ \lambda_0 $ and $ \mathcal L $.
The Jensen formula (see for instance \cite[\S 3.61]{Tit})
then gives a bound on the number of zeros of $ \det ( I + K ( \lambda, 
\lambda_0 ) )$ in $ \Omega_3 $. That 
proves the first bound in \eqref{eq:res_bound}. 

We can write 
\[ \det ( I + K ( \lambda_0  , \lambda ) ) = e^{ g( \lambda ) } \prod_{\zeta \in \Res ( P )\cap \Omega_3  } ( \lambda - \zeta ) , \]
where $ g ( \lambda ) $ is holomorphic in $ \Omega_3 $. From the 
upper bound \eqref{eq:detKla} and the lower bound \eqref{eq:detKla0} 
we conclude that $ | g ( \lambda ) | \leq C_5 $ in a smaller disc 
containing $ \Omega_2 $, with $ C_5 $ depending only on the 
previous constants. (For instance we can use the Borel--Carath\'eodory
inequality -- see \cite[\S 5.5]{Tit}.)
Hence
\[ | \det ( I + K ( \lambda_0 , \lambda ) | \geq e^{-C_6}  
 \prod_{\zeta \in \Res ( P )\cap \Omega_2  } | \lambda - \zeta | , \ \ \lambda \in \Omega_1 , \]
To deduce the 
the second bound in \eqref{eq:res_bound} from this 
we use the 
inequality 
\[  \| ( I + A )^{-1} \| \leq \frac{ \det ( I + | A | ) }{ |\det ( I + A ) | } \]
which gives
\[ \begin{split} 
\| \indic_{r \leq R} R ( \lambda ) \indic_{ r \leq R } \| & = \| 
\indic_{r \leq R }  Q ( \lambda, \lambda_0 ) \chi_3 ( I + K( \lambda, \lambda_0 ) )^{-1} \indic_{r \leq R } \| \\
& \leq \| \indic_{r \leq R }  Q ( \lambda, \lambda_0 ) \chi_3 \| 
| \det ( I + | K ( \lambda , \lambda_0 ) ) | | \det ( I + 
K ( \lambda , \lambda_0 ) )|^{-1} \\
& \leq C_7 e^{ \| K ( \lambda , \lambda_0 ) \|_{\tr} } 
\prod_{\zeta \in \Res ( P )\cap \Omega_2  } | \lambda - \zeta |^{-1} , 
\end{split} \]
for $ \lambda \in \Omega_1 $ and $ C_7 $ depending only on $ \Omega_j$'s
$ \mathcal L $, and $ R $. This completes the proof.
\end{proof}

Before proving Theorem \ref{t:q2r} we will use the construction 
of the meromorphic continuation in the proof of Proposition \ref{p:reso}
to give a general condition for smoothness of a family of resonances
(see also \cite{St3}):
\begin{equation}
\label{eq:condsm}
 ( P ( t ) - \lambda_0^2 )^{-1} \in \CI ( (-t_0, t_0 ) ; 
 \mathcal L ( L^2, L^2 ) ) , \ \ \Im \lambda_0 > 0 .
\end{equation}
That is the only property used in the proof of 
\begin{prop}
\label{p:smooth} Let $ P ( t ) $ be the family of unbounded 
operators on $ L^2 $ (of a fixed metric graph) defined by \eqref{eq:Poft}. Let $ R ( \lambda, t ) $
be the resolvent of  $P ( t ) $ meromorphically continued to $ \mathbb C$. Suppose
that $ \gamma  $ is a smooth Jordan curve such that $ R ( \lambda, t ) 
$ has no poles on $ \gamma $ for $ |t | < t_0 $. Then for $ \chi_j \in 
\CIc $, $ j=1,2 $, 
\begin{equation}
\label{eq:Rzett} \int_\gamma  \chi_1 R ( \zeta , t ) \chi_2  d \zeta \in 
\CI ( ( - t_0, t_0 ) ; \mathcal L ( L^2 , L^2 ) ) . \end{equation}
In particular, if $ \lambda_0 $ is a simple pole of $ R ( \lambda, 0 ) $
then there exist smooth families $ t \mapsto \lambda ( t ) $ and
$ t \mapsto u ( t ) \in \mathcal D_{\loc} ( P ( t ) ) $ such that
$ \lambda ( 0 ) = \lambda_0 $, $ \lambda ( t ) \in \Res ( P ( t ) )$
and $ u ( t ) $ is a resonant state of $ P ( t ) $ corresponding to $ \lambda ( t )$.
\end{prop}
\begin{proof} 
The proof of \eqref{eq:Rzett} under the condition \eqref{eq:condsm}
follows from \eqref{eq:resform} and the definitions of $ Q ( \lambda, 
\lambda_0 )$ and $ K ( \lambda, \lambda_0 )$. From that the conclusion
about the deformation of a simple resonance is immediate -- see 
\cite[Theorems 4.7,4.9]{res}. 

It remains to establish \eqref{eq:condsm}. Suppose $ f \in L^2 $ and
define $ u ( t ) := R ( \lambda_0 , t )  f \in L^2 $.
Formally, $ \dot u := \partial_t u ( t ) $ satisfies \eqref{eq:dotu2} 
with $ \dot z = 0 $ and $ z = \lambda^2_0 $. We can find a smooth
family $ g ( t) \in L^2 $ satisfying \eqref{eq:gioconda} with
$ u = u ( t ) $. We then have $ \partial_t u ( t ) = g + 
R ( \lambda_0 , t ) ( - 2 \partial_t a ( t ) u ( t ) - G ( t)  )$, 
where $ G_m := ( -e^{-2a ( t ) } \partial_x^2 - \lambda_0^2 ) g_m ( t ) $. 
By considering difference quotients a similar argument shows that
$ u ( t ) \in L^2$ is differentiable. The argument can be iterated
showing that $ u ( t) \in \CI ( (-t_0, t_0 ) , L^2 )$ and that proves
\eqref{eq:condsm}.
\end{proof}

We now give
\begin{proof}[Proof of Theorem \ref{t:q2r}]
We proceed by contradiction by assuming that, for $ 0 < \delta \ll 
\rho \ll 1 > 0 $
to be chosen, 
\[  \Res ( P ) \cap ( \Omega( \rho, \delta) + D ( 0 , \delta) ) 
=  \emptyset , \  \ 
\Omega( \rho, \delta ) := [ \lambda_0 - \rho , \lambda_0 + \rho ] - i [ 0 , \delta ] \]
does not contain any resonances. Choosing pre-compact open sets, 
independent of $ \epsilon, \rho $ and $ \delta $, 
$ \Omega ( \rho, \delta ) + D ( 0 , \delta ) 
\Subset \Omega_1 \Subset \Omega_2 $ we apply Proposition \ref{p:reso}
to see that for
\begin{equation}
\label{eq:noresup}
\| \indic_{r \leq R } R ( \lambda ) \indic_{r \leq R } \| 
\leq C_2 \delta^{-C_1} , \ \ \lambda \in \Omega( \rho, \delta ) .
\end{equation}
On the other hand, the resolvent estimate in the physical half-plane
$ \Im \lambda > 0 $ and the fact that $ \lambda_0 \in I \Subset 
( 0 , \infty ) $, give 
\begin{equation}
\label{eq:resoup}
\| \indic_{r \leq R } R ( \lambda ) \indic_{r \leq R } \| 
\leq C_3 / \Im \lambda , \ \ \ \Im \lambda > 0 , \ \  | \Re \lambda - 
\Re \lambda_0 | < \rho . 
\end{equation}
To derive a contradiction we use the following simple lemma:
\begin{lemm}
\label{l:threeline}                                                           
Suppose that $ f ( z ) $ is holomorphic in a                                  
neighbourhood of                                                              
$ \Omega := [ - \rho , \rho ] + i[ -\delta_- , \delta_+ ] $, $ \delta_\pm > 0 $.  
Suppose that, for $M >1 $, $ M_\pm > 0 $, and $ 0 < \delta_+ \leq \delta_- < 1 $, 
\begin{equation}
\label{eq:threel1} \begin{split}
&  | f (z ) | \leq M_\pm ,  \ \ \Im z = \pm \delta_+ , \ \ | \Re z | \leq
 \rho,  \ \ \ \ \ \
  | f (z ) | \leq M   ,   \ \ z \in \Omega .
\end{split} \end{equation}
and that
$ \rho^2 > ( 1 + 2 \log M ) \delta_-^2 $.
Then 
\begin{equation}
\label{eq:threeline}
| f ( 0 ) | \leq e M_+^{\theta }
M_-^{1-\theta} ,\quad
\theta:=\frac{\delta_-}{\delta_+ + \delta_-}.
\end{equation}
\end{lemm}
\begin{proof}
We consider the following subharmonic function defined in a neighbourhood of 
$ \Omega $. To define it
we put $ m_\pm = \log M_\pm $, $ m = \log M > 0 $, $ z = x + iy $, and
\[  u ( z ) := \log | f ( x + i y  ) | - \frac{ \delta_- m_+  + \delta_+  
m_-  + y  (m_+ - m_-) }{ \delta_+ + \delta_-} - K x^2 + K y^2, \]
where $ K := 2 m/( \rho^2 - \delta_-^2) $. Then for $ \Im z = \pm \delta_\pm$, 
$ | \Re z |\leq \rho $, $ u ( z ) \leq \delta_-^2 K \leq 1 $
since we assumed $ \rho^2 > ( 1 + 2 m ) \delta_-^2 $).
 When $ | \Re z | = \rho $ then $ u ( z ) \leq 2 m  - K ( \rho^2 -\delta_-^2) \leq 0 $. The maximum principle for subharmonic functions shows that
 $ \log | f ( 0 ) | - \theta m_+ - ( 1 - \theta) m_- \leq 1 $ and that
 concludes the proof. 
\end{proof} 

We apply this lemma to 
$ f ( z ) := \langle  \indic_{r \leq R } R ( z + \lambda_0 ) \indic_{r \leq R } \varphi, \psi \rangle $, $  \varphi, \psi
\in L^2 $,
with  $ M_+ = C_3/\delta_+ $, $ M=M_- = C_2 \delta^{-C_1} $. If we show that
\begin{equation}
\label{eq:fof0} 
  | f ( 0 ) | \ll \frac 1 \epsilon \| \varphi \| \| \psi \| , 
  \end{equation}
we obtain a contradiction to \eqref{eq:quasim} by putting 
$ \psi = ( P - \lambda_0^2) u $ and $ \varphi = u $ and using the support
property of $ u $ (the outgoing resolvent is the right inverse of 
$  P - \lambda_0^2 $ on compactly supported function):
\[ \begin{split}  1 &= 
\langle   R ( \lambda_0 ) ( P - \lambda_0^2) u , u  \rangle 
= \langle  \indic_{r \leq R }  R( \lambda_0 )  \indic_{r \leq R }  ( P - \lambda_0 ) u , u \rangle \ll \frac 1 \epsilon \epsilon \ll 1 . 
\end{split} \]
For $ \gamma < 1 $ choose 
$ \gamma < \gamma_1 < \gamma_2 < \gamma_3 < 1 $ and put
\[ \rho = \epsilon^{ \gamma_1} , \ \ \delta_- = \epsilon^{\gamma_2}, \ \ 
\delta_+ = \epsilon^{ \gamma_3 } . \]
Then \eqref{eq:threeline} implies \eqref{eq:fof0} and that 
completes the proof.
\end{proof}

For completeness we also include the following proposition which would
be a converse to Theorem \ref{t:q2r} for $ \gamma = 1 $.
The more subtle higher dimensional case in the semiclassical setting
was given by Stefanov \cite{St1}.

\begin{prop}
\label{p:converse}
Suppose that $ P $ satisfies the assumptions of Theorem \ref{t:q2r}
and let $ R > 0 , \delta > 0 $. There exists a constant $ C_0 $
depending only of $ R, \delta $ and $\mathcal L $ such that for any 
$ 0 < \epsilon < \delta/2  $, 
\begin{equation}
\label{eq:converse}
\begin{split} 
& D ( \lambda_0 , \epsilon ) \cap \Res ( P ) \neq \emptyset , \ \lambda_0 > \delta 
 \ \Longrightarrow \\
& \ \ \ \ \  \exists \, u \in \mathcal H_R \cap 
\mathcal D_P  , \ \| u \|= 1 , \ \
\| ( P - \lambda_0^2 ) u \| \leq C_0 \epsilon ( \lambda_0 + \epsilon)  . \end{split}
\end{equation}
\end{prop}
\begin{proof} 
Suppose that $ \lambda $ a resonance of $ P $ with $ | \lambda - \lambda _0 | < \epsilon $
and let $ v $ be the corresponding resonant state. Then in each
infinite lead, $ v_m ( x ) = a_m e^{ i \lambda x } $, $ 1 \leq m \leq K $.
As in \eqref{eq:Imz1}, 
\begin{equation*}
\begin{split} 
\Im (\lambda^2) \| v \|_{ \mathcal H_0 }^2 & = - \Im  \sum_{ m=1}^{K} \partial_x v_m ( 0 ) \bar v_m (  0  ) = - \Im \sum_{ m = 1}^K i \lambda | a_m|^2  \\
&  = 
- \Re \lambda \sum_{m=1}^K | a_m|^2 ,
\end{split}
\end{equation*}
and since $ \Re \lambda > \delta/2 > 0 $, 
$ \sum_{ m=1}^K | a_m |^2 = 2 | \Im \lambda | \| v \|_{ \mathcal H_0 } \leq  
2 \epsilon \| v \|_{ \mathcal H_0 } $.

Suppose $ r < R/2  $ and $ \chi \in \CIc ( [ 0 , 2) $ is equal to $ 
1 $ on $ [ 0 , 1 ] $. We then define 
$ \widetilde u \in \mathcal H_R \cap \mathcal D_P  $ by 
\[ \widetilde u_m ( x ) :=  \left\{ \begin{array}{ll} \chi ( x /r ) v_m ( x ) , 
& 1 \leq m \leq K \\
v_m ( x ) &  K+1\leq m \leq K + M . \end{array} \right. 
\]
Now, 
\[  \| \widetilde u \|^2 = \| v \|_{\mathcal H_0 }^2  + \sum_{ m=1}^K 
|a_m|^2 \int_\RR e^{ 2 |\Im \lambda | x} \chi ( x / r )^2 dx = 
\| v \|_{ \mathcal H_0}^2 ( 1 + \mathcal O ( \epsilon r e^{ 2 \epsilon r } ) ) , \]
and hence, 
\[ \begin{split} 
\| ( P - \lambda_0^2 ) \widetilde u \|^2 & = 
| \lambda^2 - \lambda_0^2 |^2 \| \widetilde u \|^2 + \| 
[ P , \chi ( \bullet/ r ) ] \widetilde u \|^2 \\
& \leq  ( 2 \epsilon ( \lambda_0 + \epsilon ))^2 \| \widetilde u \|^2 
+  C \sum_{ m=1}^K |a_m|^2 ( r^{-2} +  (\lambda_0 + 
\epsilon )^2 )   e^{ 2 \epsilon r } \\ 
&\leq  C_{ r ,\delta  }  \epsilon^2 ( \lambda_0 + \epsilon )^2 
\| v \|_{ \mathcal H_0 }^2 . 
\end{split} \]
We conclude that we can take $ u := \tilde u / \| \tilde u \| $ as
the quasimode.
\end{proof}

\def\arXiv#1{\href{http://arxiv.org/abs/#1}{arXiv:#1}}

\end{document}